\documentclass[aps,pra,twocolumn,superscriptaddress]{revtex4}

\usepackage{amssymb}
\usepackage{amsmath}
\usepackage{amsthm}
\usepackage{todonotes}

\usepackage[breaklinks]{hyperref}

\newcommand{\ket}[1]{| #1 \rangle}
\newcommand{\bra}[1]{\langle #1 |}

\newtheorem{theorem}{Theorem}
\newtheorem{lemma}{Lemma}
\newtheorem{corollary}{Corollary}

\begin{document}

\title{Entanglement classes of permutation-symmetric qudit states:
  symmetric operations suffice}

\author{Piotr Migda{\l}}
\affiliation{ICFO--Institut de Ci\`{e}ncies Fot\`{o}niques, 08860
  Castelldefels (Barcelona), Spain}
\email{piotr.migdal@icfo.es}
\email{pmigdal@gmail.com}
\homepage{http://migdal.wikidot.com/en}

\author{Javier Rodriguez-Laguna}
\affiliation{ICFO--Institut de Ci\`{e}ncies Fot\`{o}niques, 08860
  Castelldefels (Barcelona), Spain}
\affiliation{Mathematics Department, Universidad Carlos III de Madrid,
  Spain}

\author{Maciej Lewenstein}
\affiliation{ICFO--Institut de Ci\`{e}ncies Fot\`{o}niques, 08860
  Castelldefels (Barcelona), Spain}
\affiliation{ICREA--Instituci\'o Catalana de Recerca i Estudis
  Avan\c{c}ats, Lluis Companys 23, 08010 Barcelona, Spain}

\date{August 16, 2013}

\begin{abstract}
We analyze entanglement classes for permutation-symmetric states for
$n$ qudits (i.e., $d$-level systems), with respect to local unitary
operations (LU equivalence) and stochastic local operations and
classical communication (SLOCC equivalence). In both cases, we show
that the search can be restricted to operations where the same local
operation acts on all qudits, and we provide an explicit construction
for it.  Stabilizers of states in the form of one-particle operations
preserving permutation symmetry are shown to provide a coarse-grained
classification of entanglement classes. We prove that the Jordan form
of such one-particle operators is a SLOCC invariant. We find, as
representatives of those classes, a discrete set of entangled states
that generalize the GHZ and W states for the many-particle qudit
case. In the later case, we introduce {\em excitation states} as a
natural generalization of the W state for $d>2$.

\end{abstract}

\maketitle

\section{Introduction}

\subsection{Entanglement of multipartite pure states}

Entanglement is perhaps the most important resource for quantum
information (for a review see \cite{Horodecki2009}), and its
characterization is one of the most important tasks of quantum
theory. Particularly difficult is the problem of characterization of
entangled mixed states (for a recent review of various necessary
criteria see \cite{Guehne2009}). It might seem that the problem of
characterization of pure state entanglement is much simpler and
more tractable, but even this statement is, generally speaking, not true,
except for the case of bipartite entanglement, where the Schmidt
decomposition provides a method of classification of pure entangled
states of two parties \cite{Horodecki2009}. In a multipartite scenario
very little is known about the different classes of
entanglement. Typical questions that one would like to answer concern
entanglement classes of pure states which are invariant with respect
to local operations. The latter are typically assumed to belong to a
group (unitary, general linear, etc.). The corresponding classes of
states are called then LU equivalent, SLOCC equivalent, etc., where LU denotes
local unitary, and SLOCC --- stochastic local operations and classical
communications. Only a few rigorous results are known concerning these
questions, which we list below

\begin{itemize}
\item For three qubits a generalization of the Schmidt decomposition
  has been formulated (see \cite{acin2000generalized} and references
  therein) --- this result provides a classification of invariant
  states with respect to local unitaries. There is a considerable
  amount of work regarding this and the related problem of geometrical
  invariants by Sudbery and coworkers \cite{Carteret2000, Williamson2011}. 
\item Classification of entanglement of three qubit states according
  to LU operations and SLOCC has been presented in
  Ref. \cite{Sudbery2000} and \cite{Dur2000}, respectively.  
\item Classification of entanglement of 4 qubits according to SLOCC
  has been presented in Ref. \cite{Verstraete2002} (see also the
  papers by Miyake and Wadate \cite{Miyake2002} and by Miyake
  \cite{Miyake2003}). 
\item For many qudits a multiparticle generalization of the Schmidt
  decomposition \cite{Carteret2000, Verstraete2003} provides a general
  way to determine whether two states are LU equivalent.
\item SLOCC equivalence can be reduced to LU equivalence of the
  critical points of the so-called {\em total variance function}
  \cite{Sawicki.12}. 
\end{itemize}

There is also a considerable amount of work on many qubits states
(see \cite{Miyake2004, Miyake2004a,Kraus2010a, Kraus2010b}), but very
little is known about general many qudit states. The difficulty of
classifying entanglement for multipartite pure states was evidently
one of the motivations to look at restricted families of states. Such
restrictions are typically introduced by considering symmetries
\cite{Sawicki2012}, which might be physically motivated. In this
spirit many authors considered totally (permutation) symmetric pure
states of $n$ qubits
(see \cite{Aulbach2010,Devi2010,Mathonet2010,Aulbach2011,Cenci2010a,
  Ganczarek2011}), since such states naturally describe systems of
many bosons, and appear frequently in the context of quantum
optics. Similarly, quantum correlations in totally antisymmetric
states (as representative states of fermions) have been intensively
investigated (for a review see \cite{Eckert2002} and references
therein). In the next introductory section we focus on symmetric
states and their particular role in physical applications.

\subsection{Permutationally symmetric pure states}

A many-qudit wavefunction can be permutation-symmetric for two
reasons. One is when it describes a system of bosons, so that the
particles are indistinguishable on a fundamental level. Second is when
the particles are distinguishable but, because of a particular setting
(e.g. a Hamiltonian for which the particles form an eigenstate), they
happen to be in a permutation-symmetric state. The latter situation
occurs for instance for the Lipkin-Meshkov-Glick model
\cite{Lipkin1965} of nuclear shell structure, and related models of
quantum chaos \cite{Gnutzmann2000}. It is worth stressing that the two
situations are {\em not the same}. In the later case we are able to
manipulate each particle separately in a different way, while in the
first we are restricted to operations modifying each boson in the same
way. The question is whether those two settings give rise to the same
entanglement classes, i.e. if for symmetric states classification can
be reduced to studying operations that act in the same way on all
particles.  Moreover, entanglement geometry of permutation-symmetric
states is interesting and relevant, e.g. for quantum computation using
linear optics \cite{Aaronson2010}.  As mentioned above, this question
has been widely studied in the qubit case \cite{Aulbach2010a, Devi2010,
  Ganczarek2011}, but most of the results are not-applicable
for qudit systems of dimension $d>2$, a general case which we are
going to address.

In this paper we consider two types of equivalence ---under local
unitary operations (LU equivalence) \cite{Kraus2010a, Kraus2010b} and under
positive-operator-valued measures with post-selection. The
second is called stochastic local operations and classical
communication and is equivalent to multiplication by
invertible matrices (not necessarily unitary or even normal or
diagonalizable) on each particle.

For example, for the simplest case of two qubits, the LU-equivalence
classes are distinguished by their Schmidt coefficients, where a
unique representative can be written as $\lambda \ket{00} +
\sqrt{1-\lambda^2} \ket{11}$, with $0 \leq \lambda \leq 1/2$. On the
other hand, for SLOCC-equivalence there are only two distinct
symmetric states --- a product state $\ket{00}$ and the
Greenberger-Horne-Zeilinger (GHZ) state
$(\ket{00} + \ket{11})/\sqrt{2}$. In general, entanglement classes of
more than two particles are much more involved, even for the symmetric
qubit ($d=2$) states with three \cite{Dur2000} or four
\cite{Verstraete2002} particles.

In this paper we present two results. The first one is that, when
testing whether two permutation-symmetric $n$ qudit states are
equivalent under local transformations, the search can be restricted
to operators which act in the same way on every particle. This
property was proven for qubits \cite{Bastin2009, Mathonet2010,
  Bastin2010} in the SLOCC variant (although the unitary version can be
deduced from that proof).  For a general qudit system has remained so
far an open question \cite[Sec. 5.1.1.]{Aulbach2011}.  That is, in the
course of this paper, we prove the following:

\begin{theorem}\label{thm:main}
Let us consider two permutation-symmetric states of $n$ qudits
(i.e. $d$-level particles), $\ket{\psi}$ and $\ket{\varphi}$, for
which there exist invertible $d\times d$ matrices $A_1, \ldots, A_n$
such that
\begin{equation}
A_1 \otimes A_2 \otimes \ldots \otimes A_n \ket{\psi} =
\ket{\varphi}\label{eq:a1a2an}.
\end{equation}
Our result implies that then there exists an invertible $d\times d$
matrix $A$ such that
\begin{equation}
A \otimes A \otimes \ldots \otimes A \ket{\psi} = \ket{\varphi}.\label{eq:aaa}
\end{equation}

If we restrict ourselves to unitary matrices $A_1, \ldots, A_n$, then $A$ is
unitary.
\end{theorem}

For $A_i$ unitary, \eqref{eq:a1a2an} is a condition of the equivalence
of states under reversible local operations (or LU-equivalence), which
is proven to be the same as equivalence with respect to local
operations and classical communication \cite{Bennett2000, Vidal2000}.
Moreover, in both cases we provide a direct construction for $A$ as a
function of $A_1,\ldots, A_n$.

Our second result stems from the consideration of stabilizers of
states \cite{Cenci2010a} in the form of a matrix $B$ acting on one
particle, and its inverse $B^{-1}$ acting on another one. Only for
very specific states are there such $B$, that are non-trivial. We
show that the Jordan form of $B$, disregarding the values of the
eigenvalues, is an invariant for SLOCC-equivalence, and analyze it in
detail, providing a coarse-grained classification of the relevant
entanglement classes. If each block of the Jordan form of $B$ has a
distinct eigenvalue, then there is a unique stabilized state, up to
local operations. In particular, we find as entanglement class
representatives a $d$-level generalization of the $n$-particle GHZ
state

\begin{align}
\frac{\ket{0}^n + \cdots + \ket{d-1}^n}{\sqrt{d}}
\end{align}
and one possible generalization of the W state for $d>2$, i.e. a state
with all single particle state indices adding up to $d-1$, that is
\begin{align}
{\textstyle {n + d - 2 \choose d - 1}^{-1/2}} 
\sum_{i_1+\cdots+i_n = d-1} \ket{i_1} \ket{i_2}\cdots \ket{i_n},
\label{eq:excitationnormalized}
\end{align}
which we call an \textit{excitation state}.

For two particles both classes coincide, as e.g.
\begin{align}
\ket{00} + \ket{11} + \ket{22} \cong \ket{02} + \ket{11} + \ket{20}
\end{align}
Table \ref{tab:n3} summarizes the entanglement classes related to
Jordan blocks for the simplest non-trivial case, i.e. $n=3$ particles
(a general construction is in \eqref{eq:unique_general}).  We adopt a
special notation for the Jordan block structure.  Outer brackets
separate eigenspaces with different eigenvalues, while the inner
brackets separate different Jordan blocks of the same eigenvalue.
Each number is the dimension of a single Jordan block.  Ordering of
the terms does not matter, in either the inner or outer brackets.  For
example: $\{ \{ 2 \}\}$ is a matrix with only one Jordan block, $\{ \{
1, 1 \}\}$ is proportional to the identity matrix and $\{ \{ 1 \}, \{
1 \} \}$ is a matrix with two different eigenvalues, that is
($\lambda_1 \neq \lambda_2$):

\begin{align}\label{eq:block-notation}
\{ \{ 2 \}\} & \equiv
\left[
\begin{array}{cc}
\lambda_1 & 1\\
0 & \lambda_1
\end{array}
\right]\\
\{ \{ 1, 1 \}\} & \equiv
\left[
\begin{array}{cc}
\lambda_1 & 0\\
0 & \lambda_1
\end{array}
\right]\\
\{ \{ 1 \}, \{ 1 \} \} & \equiv
\left[
\begin{array}{cc}
\lambda_1 & 0\\
0 & \lambda_2
\end{array}
\right]
\end{align}
 
The number of different Jordan block structures for a given $d$ is
given by double partitions~\cite{oeisA001970}.

\begin{table}
\begin{tabular}{|l|l|l|}

\hline
$d$ &
Block structure &
A class representative\\

\hline
2 &
$\{ \{ 2 \}\}$ &
$\ket{001} + \ket{010} + \ket{001}$\\

\hline
 &
$\{ \{ 1, 1 \}\}$ &
(not unique) any state\\

\hline
 &
$\{ \{ 1 \}, \{ 1 \} \}$ &
$\ket{000} + \ket{111}$\\

\hline
3 &
$\{ \{ 3 \} \}$ &
$\ket{002} + \ket{020} + \ket{200}$ \\

 &
 &
$+ \ket{011} + \ket{101} + \ket{110}$\\

\hline
 &
$\{ \{ 2, 1 \} \}$ &
(not unique)\\

\hline
 &
$\{ \{ 1, 1, 1 \} \}$ &
(not unique) any state\\

\hline
 &
$\{ \{ 2 \}, \{ 1 \} \}$ &
$\ket{001} + \ket{010} + \ket{100} + \ket{222}$\\

\hline
 &
$\{ \{ 1, 1 \}, \{ 1 \} \}$ &
(not unique)\\

\hline
 &
$\{ \{ \{ 1 \}, \{ 1 \}, \{ 1 \} \} \}$ &
$\ket{000} + \ket{111} + \ket{222}$\\

\hline

\end{tabular}
\caption{A summary of entaglement classes related to Jordan blocks,
  for the case of three qubits ($d=2$) and quitrits ($d=3$). The
  notation is explained in the main text
  \eqref{eq:block-notation}. The general construction for the unique
  states is given in \eqref{eq:unique_general}.}
\label{tab:n3}
\end{table}

This paper is organized as follows: Sec. \ref{symmetry} proves that it
is sufficient to study invariance under symmetric
transformations. Section \ref{classes} discusses the entanglement
classes which can be obtained by studying stabilizer operators related
to one-particle transformations. Section \ref{s:conclusion} is devoted
to conclusions and further work.

\section{Symmetric operations suffice}
\label{symmetry}

We start with an approach similar to the one from \cite{Mathonet2010}.
Let us consider two permutation-symmetric states, $\ket{\psi}$ and
$\ket{\varphi}\in {\cal S}$, with ${\cal S}$ denoting the symmetric
subspace of the full Hilbert space.  If \eqref{eq:a1a2an} holds, then
any different permutation of $A_1, \cdots, A_n$ will also work. In
order to show this property explicitly, we may use $\ket{\psi} =
P_\sigma \ket{\psi}$ and $\ket{\varphi} = P_{\sigma^{-1}}
\ket{\varphi}$, where $P_\sigma$ is a permutation matrix for the
permutation of particles $\sigma$, i.e. 
$P_\sigma \ket{i_1i_2 \cdots i_n}=\ket{i_{\sigma(1)} i_{\sigma(2)}
\cdots i_{\sigma(n)}}$.

Since all $A_i$ are invertible, it means in particular that
\begin{equation}
\left(A_2^{-1} \otimes A_1^{-1} \otimes \cdots \otimes A_n^{-1} \right)
\left(A_1 \otimes A_2 \otimes \cdots \otimes A_n \right) \ket{\psi} = \ket{\psi}
\end{equation}
or, equivalently, 
\begin{equation}
 \left(B \otimes B^{-1} \otimes \mathbb{I} \otimes \cdots \otimes \mathbb{I}
\right) \ket{\psi} = \ket{\psi},\label{eq:bbinv}
\end{equation}
where $B=A_2^{-1} A_1$.

From now on, we will use a subscript in parentheses to indicate the
position of an operator in the tensor product, e.g.,
\begin{equation}
 B_{(2)} \equiv \mathbb{I}  \otimes B \otimes \mathbb{I} \otimes \cdots \otimes
\mathbb{I},
\end{equation}
where the total number of factors is $n$.

First, let us show that if an operation on one particle can be
reversed by applying the inverse operation on a {\em different}
particle, then that single-particle operation must preserve
permutation symmetry of the state.

\begin{lemma}\label{thm:bbinv2b}
For a symmetric $\ket{\psi} \in \mathcal{S}$ the equality
\eqref{eq:bbinv} 
\begin{equation}
B_{(1)} B_{(2)}^{-1} \ket{\psi} = \ket{\psi}
\end{equation}
holds if and only if 
\begin{equation}
B_{(1)} \ket{\psi} \in \mathcal{S}. \label{eq:b}
\end{equation}
\end{lemma}

\begin{proof}

\begin{align}
B_{(1)} \ket{\psi} \in \mathcal{S} \Leftrightarrow B_{(1)} \ket{\psi} =
B_{(2)} \ket{\psi}\\
\Leftrightarrow B_{(1)} B_{(2)}^{-1} \ket{\psi} = \ket{\psi}
\end{align}

\end{proof}

Now we will show that the action of the aforementioned single-particle
operation $B_{(1)}$ can be expressed as an operation acting in the
same way on every particle $S^{\otimes n}$. Intuitively, we must
search for an $n$-th root of $B$. But not all such $n$-th roots will
work, as the following example shows:


Consider $S=\sigma_x$, which is a square root of $B=I$, acting on
$\ket{00}\in \cal{S}$. While $B_{(1)}\ket{00}\in \cal{S}$,
$S_{(1)}\ket{00}=\ket{10}\not\in \cal{S}$, and $S \otimes S
\ket{00}=\ket{11}$, which, despite being symmetric, is not the desired
state. Thus, the relevant question is: {\em which one is the
  appropriate $n$-th root?}

Before we can proceed, we need a few lemmas.

\begin{lemma}
\label{thm:symcom}
If $\ket{\psi} \in {\cal S} $, $X_{(1)}\ket{\psi}\in {\cal S}$ and
$Y_{(2)}\ket{\psi}\in{\cal S}$, then $X_{(1)}Y_{(2)}\ket{\psi}\in
{\cal S}$ $\Leftrightarrow$ the commutator acting on the state
vanishes $[X_{(1)}, Y_{(1)}]\ket{\psi}=0$. 
\end{lemma}
 
\begin{proof}
If the final state is symmetric, we may permute the first two
particles without altering the result:
\begin{align}
0 &= X_{(1)}\left(Y_{(2)}\ket{\psi}\right) - Y_{(1)}\left(X_{(2)}\ket{\psi}\right)\\
&= X_{(1)}\left(Y_{(1)}\ket{\psi}\right) -
Y_{(1)}\left(X_{(1)}\ket{\psi}\right)\\
&= [X_{(1)}, Y_{(1)}]\ket{\psi}.
\end{align}
\end{proof}

To see how commutativity is important, take as an example
\begin{equation}
X = 
\left[
\begin{matrix}
1 & 1\\
0 & 1
\end{matrix}
\right], \quad
Y = 
\left[
\begin{matrix}
0 & 1\\
1 & 0
\end{matrix}
\right].
\end{equation}
acting on $\ket{\psi}=(\ket{01}+\ket{10})/\sqrt{2}$ (i.e. $\ket{\psi}$,
$X_{(1)}\ket{\psi}$ and $Y_{(2)}\ket{\psi}$ are symmetric, but $X \otimes Y
\ket{\psi} = (\ket{00}+\ket{01}+\ket{11})/\sqrt{2}$ is not).

\begin{lemma}\label{thm:comm3}
Moreover, for $n\geq3$ the commutator acting on the state always
vanishes, i.e.  $[X_{(1)}, Y_{(1)}]\ket{\psi}=0$.
\end{lemma}

\begin{proof}
\begin{align}
&X_{(1)} Y_{(1)} \ket{\psi}
= X_{(1)} Y_{(3)} \ket{\psi}
= X_{(2)} Y_{(3)} \ket{\psi}\\
&= Y_{(3)} X_{(2)} \ket{\psi}
= Y_{(1)} X_{(1)} \ket{\psi}
\end{align}
\end{proof}

\begin{lemma}
If $X_{(1)}\ket{\psi}$ is symmetric, then $X^p_{(1)}\ket{\psi}$ is
symmetric for all natural $p$ (integer if $X$ is invertible).
\end{lemma}

\begin{proof}
We use mathematical induction with respect to $p$, starting at
$p=0$. Since $X$ commutes with $X^p$ (even without the restriction to
a specific state), then using Lemma \ref{thm:symcom},
$X^p_{(1)}\ket{\psi}\in{\cal S}$ implies that
$X^{p+1}_{(1)}\ket{\psi}\in {\cal S}$. If $X$ is invertible, we may
use the same argument for $X$ and $X^{-p}$, respectively.
\end{proof}

\begin{corollary}
Moreover, we get
\begin{align}
X^p_{(1)}\ket{\psi} = X^{p_1}\otimes X^{p_2}\otimes \cdots \otimes X^{p_n}
\ket{\psi},
\end{align}
for any integers $p_i$ (can be negative if $X$ is invertible) that add up to
$p$.
\end{corollary}

\begin{corollary}
\label{corol:analytic}
In particular, $f(X)_{(1)}\ket{\psi}\in {\cal S}$ for any analytic
function $f(z)$.
\end{corollary}

\begin{theorem}
\label{thm:sssb}
For any $X$ and $\ket{\psi}\in {\cal S}$ it holds that if
\begin{equation}
X_{(1)}  \ket{\psi} = \ket{\phi} \in {\cal S}
\end{equation}
then there exists a single-particle operator $S$ such that $S^n=X$,
$S_{(1)}\ket{\psi}\in{\cal S}$ and
\begin{equation}
 \left(S \otimes S \otimes \cdots \otimes S \right)  \ket{\psi} =
\ket{\phi}.
\end{equation}
\end{theorem}

\begin{proof}
The proof outline is the following: we prove that, among the $n$-th
roots of operator $X$, there is (at least) one, $S$, which {\em can be
  expressed as a polynomial} of $X$; following Corollary
\ref{corol:analytic}, we get $S_{(1)}\ket{\psi}\in\cal{S}$ and the
rest of the theorem follows.

The $n$-th root function is multivalued, so we can not use it to prove
the theorem as it stands. Let us, then, prove that there exists a
polynomial function $f$, such that $[f(X)]^n=X$.

Let $\left\{\lambda_i\right\}$ be the eigenvalues of $X$, with
algebraic multiplicities $\left\{m_i\right\}$ (i.e. the size of the
largest Jordan block related to such eigenvalue). Matrix function
theory \cite[Chapter 1]{Higham2008} states that the action of any
analytical function $f$ on a matrix $X$ is completely determined by
the set of values $\left\{f(\lambda_i)\right\}$, along with the
derivatives $\left\{ f^{(k)}(\lambda_i)\right\}$, up to degree
$m_i$. Let us choose, for each $i$ separately, $f(\lambda_i)$ and
$f^{(k)}(\lambda_i)$ from the same branch of the complex $n$-th root
function. It is always possible to find a polynomial $f$ that takes
exactly those values and derivatives at the eigenvalues of $X$, e.g.,
via Hermite interpolation.  Thus, we can define $S\equiv f(X)$, and we
have $S^n=X$, as required.

Combining this result with the corollaries, we get that
\begin{align}
S^{\otimes n} \ket{\psi} = S^{n}_{(1)} \ket{\psi} = X_{(1)} \ket{\psi}.
\end{align}

\end{proof}

The converse of theorem \ref{thm:sssb} is false. Take, e.g.,
$S=\sigma_x$ and $\ket{\psi}=\ket{00}$. It is true that $S \otimes S
\ket{00}=\ket{11}\in \cal{S}$, yet there is no $B$ such that
$B_{(1)}\ket{00}=\ket{11}$.

\subsection{Explicit formula for symmetrization}

In this section we provide the explicit form of $A$, given all $A_i$.
Let $B_{ij} \equiv A_i^{-1} A_j$. Thus, operator $B$ in the previous
section would correspond to $B_{12}$ with the new
notation. Transforming \eqref{eq:a1a2an} we get
\begin{align}
&A_1 \otimes A_2 \otimes \cdots \otimes A_n\ket{\psi}\\
&= A_1 \otimes A_1 B_{12} \otimes \cdots \otimes A_1 B_{1n} \ket{\psi}
\\
&= A_1^{\otimes n} B_{12\ (2)} B_{13\ (3)} \cdots B_{1n\ (n)}\ket{\psi}.
\end{align}

All $B_{1j\ (j)}\ket{\psi}$ are symmetric states, similarly to
$B_{(1)}\ket{\psi}$. Consequently, the last part can be reshuffled as
\begin{align}
A_1^{\otimes n} \left(B_{11}B_{12}\cdots B_{1n} \right)_{(1)} \ket{\psi}.
\end{align}
Note that no requirements are imposed about their commutativity.
Using Lemma \ref{thm:sssb} we get $A = A_1 S$, where $S$ is an
appropriate $n$-th root of $B_{11}B_{12}\cdots B_{1n}$.

Moreover, when all $A_i$ are unitary, then $S$ is unitary, since roots
of unitary matrices can be chosen to be unitary given that $f(U X
U^\dagger) = U f(X) U^\dagger$ for all unitary $U$.  This finalizes
the proof of Theorem~\ref{thm:main}.

\section{Symmetry classes from single-particle stabilizers}
\label{classes}

A well-known strategy in the search for entanglement classes is to
study the dimension of the stabilizers of a state \cite{Cenci2010a},
i.e.: operators $X$ such that $X \ket{\psi} = \ket{\psi}$. In our case
it is natural to consider stabilizers in the form of $X = B_{(1)}
B_{(2)}^{-1}$, and state that {\em $B$ stabilizes
  $\ket{\psi}\in\cal{S}$} as a convenient shorthand
notation. Following Lemma \ref{thm:bbinv2b}, $B$ stabilizes
$\ket{\psi}\in\cal{S}$ if and only if $B_{(1)}\ket{\psi}\in\cal{S}$.
Bear in mind that a set of $B$ stabilizing a particular state is
guaranteed to form a group only for $n\geq3$, as follows from Lemma
\ref{thm:comm3}.

Let us consider the Jordan normal form $J$ of $B$. We have shown that
all local operations for symmetric states are equivalent to the action
of the same single-particle operation on all qudits: $A^{\otimes n}$.
Consequently, if a state is stabilized by $B$, a SLOCC transformed
state is stabilized by some $A B A^{-1}$, i.e.: {\em the Jordan form
  of the stabilizer is preserved}.

Below, we prove the following facts relating the Jordan form of $B$ to
the stabilized states. First, we show that the precise eigenvalues are
not important --- only their degeneracies matter (see the notation
from Table \ref{tab:n3}). We show, second, that stabilized states
never mix eigenspaces of different eigenvalues. In particular, this
means that the problem can be split into a direct sum over distinct
eigenvalues. Third, we show the explicit form of a state stabilized by
a single Jordan block.  Fourth, we show that when eigenvalues are
non-degenerate, there is a unique state related to it (up to SLOCC).
Fifth, we proceed to write down states for multiple Jordan blocks with
the same eigenvalue. This will complete the characterization of states
stabilized by any $B$.

\begin{theorem}
The set of states stabilized by $B$ does not depend on the particular values of
its eigenvalues, as long as (non-)degeneracy is preserved.
\end{theorem}

\begin{proof}
We will show that mapping eigenvalues to different ones does not break
the stabilizer's condition.  Let us choose a complex function $f(z)$
such that (i) for all eigenvalues $f(\lambda_i) = \tilde{\lambda}_i$,
and (ii) $f^{(k)}(\lambda_i)=\delta_{0k}$ for all $k$ up to the
algebraic multiplicity of each $\lambda_i$. Now, $f(B)_{(1)}$ is also
a stabilizer of $\ket{\psi}$, with the same Jordan blocks, but
arbitrarily set eigenvalues.
\end{proof}

In particular, for $d=2$ the only two non-trivial Jordan forms of $B$
are related to the GHZ state (two different eigenvalues) and the W
state (single eigenvalue). We proceed to show that stabilized states
never mix subspaces with different eigenvalues.

Given a subspace $V$, let us denote by $\mathrm{Sym}^n(V)$ the
permutation-symmetric subspace of $V^{\otimes n}$.  Then we have the
following:

\begin{theorem}
For a given Jordan form $J$ with generalized eigenspaces $V_1, \cdots,
V_p$ for distinct eigenvalues, stabilized states are of the form
\begin{align}
\ket{\psi} \in \bigoplus_{i=1}^p \mathrm{Sym}^{n}(V_i).
\end{align}
That is, they contain no vectors mixing components from Jordan blocks
of different eigenvalues.
\end{theorem}

\begin{proof}
Let $\ket{\mu}$ and $\ket{\nu}$ be one-particle states
($\mu,\nu\in\{0,\cdots,d-1\}$) that belong to blocks of $J$ with
different eigenvalues.  Let us take $f(B)$ mapping all subspaces to
zero, except the one to which $\ket{\mu}$ belongs, which we map to 1.
Suppose that $\ket{\psi}$ has a component containing a product of
$\ket{\mu}$ and $\ket{\nu}$ (at different sites). Then, in particular,
it has $\ket{\mu}\ket{\nu}\ket{\xi}$ and
$\ket{\nu}\ket{\mu}\ket{\xi}$, for some symmetric $\xi$ (perhaps
containing $\ket{\mu}$ or $\ket{\nu}$ as well).

But
\begin{align}
&f(B)_{(1)} \left(\ket{\mu}\ket{\nu}\ket{\xi} +
\ket{\nu}\ket{\mu}\ket{\xi}\right)\\
&=  \ket{\mu}\ket{\nu}\ket{\xi}.
\end{align}
The right hand side cannot be paired with any other terms in order to
make a symmetric state. So, $f(B)_{(1)} \ket{\psi}$ is not symmetric,
which contradicts the assumption. Thus, a stabilized state can not
contain a term with a product of elements from two Jordan subspaces
with different eigenvalues.
\end{proof}

Thus, when $J$ has $d$ distinct eigenvalues, the stabilized state is a
generalized GHZ state:
\begin{equation}
\ket{\psi} = \alpha_0 \ket{0}^n + \cdots + \alpha_{d-1} \ket{d-1}^n.
\end{equation}

When we consider local unitary equivalence, then the set of
$\{|\alpha_i|^2 \}_{i\in \{0,\cdots,d-1\}}$ distinguishes classes,
whereas for SLOCC, the state is equivalent to any other with the same
number of non-zero $\alpha_i$.

Now, it suffices to focus on a {\em Jordan subspace} related to a
single eigenvalue.  Still, for a single eigenvalue there may be more
than one Jordan blocks, i.e.  invariant subspaces. We start with the
analysis of a {\em single Jordan block}, and then generalize our
result to more blocks with the same eigenvalue.

\begin{theorem}
Let $K$ be a $k \times k$ Jordan block with eigenvalue zero, i.e.
$\sum_{i=1}^{k-1}\ket{i-1}\bra{i}$. Its stabilized states are
\begin{equation}
\ket{\psi} = \sum_{j=0}^{k-1} \alpha_j \ket{E_j},\label{eq:ejsum}
\end{equation}
where $\ket{E_j}$ is a symmetric state with $j$ {\em excitations},
i.e.  a symmetrized sum of all basis states for which the sum of the
particle indices is $j$:
\begin{align}
\ket{E_j} = \sum_{i_1+\cdots+i_n = j} \ket{i_1} \ket{i_2}\cdots \ket{i_n}. 
\end{align}
\end{theorem}

\begin{proof}
First, let us show that all states $K_{(1)} \ket{E_j}$ are symmetric, as long as
$j < k$.

\begin{align}
K_{(1)} \ket{E_j} &= \sum_{i_1+\cdots+i_n = j} \ket{i_1 - 1} \ket{i_2}\cdots
\ket{i_n}\\
&= \sum_{i'_1+\cdots+i_n = j - 1} \ket{i'_1} \ket{i_2}\cdots \ket{i_n} =
\ket{E_{j-1}},
\end{align}
where we use $\ket{-1} \equiv 0$ as a convenient shorthand
notation. Now we can apply a substitution $i'_1 = i_1 - 1$ and change
the summation limit (thus requiring $j < k$).

Let us now show that all stabilized states $\ket{\psi}$ have the form
of \eqref{eq:ejsum}.  We proceed by induction with respect to $n$, the
number of particles. For $n=1$ (inductive basis), all basis states are
stabilized. Now, let us assume that the condition works up to a given
$n$. As $K_{(1)}$ reduces the total number of excitations by 1, it
suffices to look at subspaces of fixed $j$.  Together with the
inductive assumption (in particular, the fact that the first $n$
particles must remain in a permutation symmetric state after
application of $K_{(1)}$) we get a general form
\begin{align}
\ket{\xi} = \sum_{l=0}^{j} \beta_l \ket{E_{j-l}} \ket{l}.
\end{align}

To find the actual constraints on $\beta_l$, we just note that the
assumed symmetry of $K_{(1)}\ket{\xi}$ implies that
\begin{align}
 K_{(1)}\ket{\xi} &= \sum_{l=0}^{j-1} \beta_l \ket{E_{j-l-1}} \ket{l}\\
  = K_{(n+1)}\ket{\xi} &= \sum_{l=1}^{j} \beta_l \ket{E_{j-l}} \ket{j-1}.
\end{align} 
Again, with a simple shift of index, and using the orthogonality of
the components, we get $\beta_l = \beta_{l+1}$. Thus, there is only
one state (up to a factor) for a given $j$ that remains symmetric
after $K_{(1)}$.

\end{proof}

When considering SLOCC-equivalence, we may take $\ket{E_{k-1}}$ as a
representative of the states stabilized by $K_{(1)}$. The reason is
that all other states (with $\alpha_{k-1}\neq0$) can be built from it
via an operator $\sum_{j=0}^{k-1} \alpha_{k-1-j} K_{(1)}^j$. This
operator is invertible, since its determinant is $\alpha_{k-1}^k$.

Throughout the derivation, we work with unnormalized states for
convenience.  The properly normalized excitation state is given in
Eq. \eqref{eq:excitationnormalized}.

\begin{theorem}
There is a unique (up to SLOCC operations) state stabilized by $B$ if
and only if each Jordan block of $B$ has a distinct eigenvalue.

An $(n>2)$-particle state $\ket{\psi}\in\cal{S}$, stabilized by $B$,
is unique (up to SLOCC) if and only if each block of its Jordan form
has a distinct eigenvalue and no other $B'$ exists with a greater
number of eigenvalues or a lesser number of Jordan blocks.
\end{theorem}

The formulation may seem complicated, but we want to exclude {\em
  degenerate} states which are also stabilized by other operators.
For the excitation state we want to ensure that the amplitude of the
$(d-1)$ excitations is non-zero (otherwise it is stabilized also by a
matrix with two eigenvalues), or, for the GHZ states, that all
amplitudes are non-zero (otherwise, two eigenvalues can be merged into
one, forming a single Jordan block).  For example, a three qutrit pure
state $\ket{000}+\ket{111}$ is stabilized by a matrix with its Jordan
block structure $\{ \{ 1 \}, \{ 1 \}, \{ 1 \} \}$ (as for the GHZ
state). However, unlike $\ket{000}+\ket{111}+\ket{222}$ (the GHZ
state), it is also stabilized by a matrix with one fewer Jordan block,
$\{ \{ 1 \}, \{ 2 \} \}$.

\begin{proof}
``$\Leftarrow$''

We have already shown that the GHZ-like state with all amplitudes
different from zero is unique, as well as the excitation state with
non-zero amplitude for the highest excitation. It follows as well for
any state without blocks of the same eigenvalue, as the problem can be
split into a problem for each eigenvalue.

If any amplitude is zero in the GHZ-like case, the state is also
stabilized by a $B$ with a Jordan block of dimension 2.

If the amplitude for the highest excitation is zero, in the excitation
state, the state is also stabilized with a $B$ with one more
eigenvalue.

``$\Rightarrow$''

If there are two blocks with the same eigenvalue, then we can take two
one-particle eigenvectors $\ket{\mu}$ and $\ket{\nu}$ having the same
eigenvalue. Let us look at the projection of $\ket{\psi}$ on the
subspace spanned by $\text{Sym}^n(\text{lin}\{\ket{\mu},
\ket{\nu}\})$. Then, in particular, a linear combination with non-zero
coefficients of elements with zero, one and two $\ket{\nu}$ states
among all other $\ket{\mu}$ does not give rise to more blocks or
eigenvalues, but gives rise to some states which cannot be
interchanged with local operations.

\end{proof}

\begin{corollary}
The number of Jordan block structures with non-degenerate eigenvalues
is the same as the number of integer partitions of
$d$~\cite{oeisA000041}.

A general construction of such a state is
\begin{align}
\bigoplus_{i=1}^{\#\mathrm{blocks}} \ket{E_{k_i}},\label{eq:unique_general}
\end{align}
where $k_i$ is the dimension of the $i$-th Jordan block, in descending
order.  In particular, for GHZ states there are only blocks of size
$k_i=1$, whereas for the excitation state there is only one block,
$k_1=d$.
\end{corollary}

It is also relevant to ask about stabilized states for $B$ whose
Jordan decomposition contains two different blocks with the same
eigenvalue.


Let us use a one-particle basis given by $\ket{i^{(b)}}$, where $i$
denotes the excitation level (i.e. the largest $i$ such that $J^i$
acting on this vector is non-zero) and $b$ the Jordan block to which
it belongs. First, we notice that the sum of the excitations $j$ in a
given state is decreased by $1$ after the action of $K_{(1)}$. Second,
we notice that the excitations can be distributed among all Jordan
subspaces which are {\em big enough} (i.e. all blocks of size strictly
smaller than $j$). Moreover, the distribution among such Jordan
subspaces needs to be permutation-invariant.

\begin{theorem}
An unnormalized state of excitation $j$ distributed among $s$ blocks
(with weights $n_1,n_2,\cdots$ adding up to $n$, related to
distribution of excitations among Jordan blocks) reads

\begin{equation}
\ket{E_j^{n_1,n_2,\cdots}} = \sum_{\vec{b}: \#i =
n_i}\sum_{i_1+\cdots+i_n=j}\ket{i_1^{(b_1)}}\cdots
\ket{i_n^{(b_n)}}.\label{eq:ej_multimode}
\end{equation}

We will show by induction that only states of the form
$\ket{E_j^{n_0,n_1,\cdots}}$ are stabilized by such $J$.

\end{theorem}

For example, one excitation $j=1$ among two particles, distributed among two
modes ($n_1=1$, $n_2=1$) reads
\begin{align}
\ket{E_1^{1,1}} &= \ket{0^{(1)}1^{(2)}} + \ket{1^{(1)}0^{(2)}}\\
&+ \ket{0^{(2)}1^{(1)}} + \ket{1^{(2)}0^{(1)}}.
\end{align}

\begin{proof}
The induction basis is for $n=1$
and holds trivially (as it works for all states). So let us assume that
\eqref{eq:ej_multimode} holds for $n$.

For $n+1$ particles, a generic state with fixed $j$ and $n_1,n_2,\cdots$ is
\begin{align}
\sum_{l=0}^j \sum_{b=1} \beta_{l,b}
\ket{E_{j-l}^{n_1-\delta_{b1},n_1-\delta_{b1},\cdots}} \ket{l^{(b)}}.
\end{align}
Applying $J_{(1)}$ and $J_{(n+1)}$ on the state above, we get a
relation $\beta_{l,b}=\beta_{l+1,b}=\beta_b$. Moreover, from the
condition of permutation symmetry for blocks (i.e. components with the
same $(b)$) we get that all $\beta$ need to be the same, so the state
is of the form \eqref{eq:ej_multimode}.
\end{proof}

This finalizes the classification of symmetric states for which
\eqref{eq:bbinv} holds.

\section{Conclusion and further work}\label{s:conclusion}

This paper proves an open conjecture regarding the classification of
pure symmetric states under local operations. We show that the study
of homogeneous operations, i.e.: those where the same single-particle
operator acts on each particle, suffices.

Furthermore, it introduces and analyzes entanglement classification by
checking which one-particle operations preserve permutation
symmetry. In that classification we obtain a sequence of states,
unique up to SLOCC. On one extreme we find the multiparticle GHZ
state, whereas on the other there is a $(d-1)$ excitation state, which
is a natural generalization of the W state resulting from the
classification scheme.

Moreover, some questions are left open:

\begin{itemize}
\item Whether invariance under all local operations (that is, not only
  invertible operations) on symmetric states can be represented as the
  same transformation for each particle.
\item Whether the application of $k$-particle transformations on
  permutation-symmetric states which are reversible by acting on another
  part will give rise to a different entanglement classification.
\end{itemize}

\subsection*{Acknowledgements}

We would like to thank Barbara Kraus and Julio de Vicente for fruitful
discussions and valuable remarks. Furthermore, we would like to thank Federico
Poloni \cite{MO122823} for pointing out \cite{Higham2008}. 

We would like to acknowledge the Spanish MINCIN/MINECO projects
FIS2008-00784 (TOQATA) and FIS2012-33642, EU integrated projects AQUTE
and SIQS and Chist-Era project DIQIP. P.M. and J.R.L. would like to
acknowledge ICMAT for hospitality.

\end{document}